\documentclass[conference, 10pt]{IEEEtran}
\usepackage{setspace}
\usepackage{amsmath}
\usepackage{amssymb}
\usepackage{mathrsfs}
\usepackage{amsthm}
\usepackage{multirow}
\usepackage{color}
\usepackage{array}
\usepackage{float}
\usepackage{url}
\usepackage{comment}
\usepackage{enumerate}
\usepackage{graphicx}
\usepackage{enumitem}
\usepackage{eucal}
\usepackage{tabularx}
\usepackage{makecell}
\usepackage{hyperref}
\usepackage{cleveref}
\usepackage{fixfoot}
\usepackage{caption}
\renewcommand{\thetable}{\arabic{table}}

\usepackage[style=ieee]{biblatex}
\usepackage{algorithm}
\usepackage[noend]{algorithmic}
\usepackage[table]{xcolor}
\usepackage{blkarray}
\usepackage{balance}

\usepackage{stmaryrd}

\usepackage[top=0.55in, bottom=0.871in, left=0.6in, right=0.621in]{geometry}

\IEEEoverridecommandlockouts

\theoremstyle{plain}
\newtheorem{theorem}{Theorem}

\theoremstyle{definition}
\newtheorem{definition}{Definition}

\newtheorem{remark}[definition]{Remark}

\newcommand{\FF}{\mathbb{F}}

\DeclareMathAlphabet{\mathbfsl}{OT1}{ppl}{b}{it} %{OT1}{cmr}{bx}{it}

\newcolumntype{Y}{>{\centering\arraybackslash}X} % for centering tables in tabularx

\title{Two-Server Private Information Retrieval with Optimized Download Rate and Result Verification\\[-3mm]}

\author{
  \IEEEauthorblockN{Stanislav Kruglik\IEEEauthorrefmark{1}, Son Hoang Dau\IEEEauthorrefmark{2}, Han Mao Kiah\IEEEauthorrefmark{1}, Huaxiong Wang\IEEEauthorrefmark{1}}
	
 \IEEEauthorblockA{\small \IEEEauthorrefmark{1} School of Physical and Mathematical Sciences, 
		Nanyang Technological University, Singapore
    }
 \IEEEauthorblockA{\small \IEEEauthorrefmark{2} School of Computing Technologies, STEM College, RMIT University, Melbourne, Australia
    }
  {\small stanislav.kruglik@ntu.edu.sg, sonhoang.dau@rmit.edu.au, hmkiah@ntu.edu.sg, hxwang@ntu.edu.sg} 
  }

\begin{document}
\date{}

%\addtolength{\topmmargin}{+0.1cm}

\maketitle

\hspace*{-12pt}
\begin{abstract}
	Private Information Retrieval (PIR) schemes allow a client to retrieve any file of interest, while hiding the file identity from the database servers. 
	In contrast to most existing PIR schemes that assume honest-but-curious servers, we study the case of dishonest servers. 
	The latter provide incorrect answers and try to persuade the client to output the wrong result. 
	We introduce several PIR schemes with information-theoretic privacy and result verification for the case of two servers.
	Security guarantees can be information-theoretical or computational, and the verification keys can be public or private. 
	In this work, our main performance metric is the download rate.
	%Motivated by the asymmetry between the number of upload and download symbols in practical scenarios, we focus on minimizing the download complexity. 
\end{abstract}

\begin{comment}
Private Information Retrieval (PIR) schemes on databases of $m$ files replicated on $n$ servers allow a client to retrieve any file of interest remaining its identity hidden from the database servers. In contrast to the majority of existing PIR that assume hones-but-curious servers, we consider the more complex case of dishonest servers. The latter can provide incorrect answers trying to persuade the client to output the wrong result. We introduce several PIR schemes with information-theoretical privacy and result verification for the extreme two-server case. The result verification can be information-theoretical or computational, as well as public or private. Motivated by the asymmetry between the number of upload and download symbols in practical scenarios, we focus on minimizing the download complexity. 
\end{comment}

\setstretch{0.97}

\section{Introduction}
A {\em private information retrieval} (PIR) scheme allows a user to retrieve a given file $x_i$ from a database $\textbf{x}^T=x_1\cdots x_m$, while keeping its {\em identity} or {\em index} $i\in[m]$ private from the database servers \cite{chor}. 
The problem is motivated by the necessity to preserve the privacy of not only the sensitive content downloaded from public databases, but also the identity of the queried record \cite{overview, overview2}. 
Examples include the price of a specific stock or a specific blockchain transaction. 
A trivial solution is to simply download the entire database and this clearly incurs tremendous communication costs. 
Unfortunately, in the case of a single server, Chor {\em et al.}~\cite{chor} showed that this was the best information-theoretically secure solution.
Nevertheless, in the same seminal paper, Chor {\em et al.}~\cite{chor} showed that when the content are replicated among several servers, the communication cost can be significantly reduced.
Following~\cite{chor}, several authors have introduced PIR schemes that progressively reduced the communications cost~\cite{Yekhanin, Yekhanin2, Dvir}. 
Formally, in this model, the client queries each of the $k$ servers (each stores $x_1\cdots x_m$) once and retrieves $x_i$, 
while keeping the index $i$ private from up to $t$ honest-but-curious servers. 
In PIR literature, such a scheme is called {\em $t$-private $k$-server PIR scheme} and such a property is known as {\em $t$-privacy}. Motivated by the large size of stored files, we ignore the upload cost. In other words, we assume queries are small compared to the downloaded file.
Hence, our main performance metric is the {\em download rate}, defined as the ratio of the retrieved file size to the amount of information downloaded by the user \cite{dr1, dr2, dr3}. 
The {\em PIR capacity} is defined as the maximum achievable download rate and the capacity is shown in \cite{SunJafar} to be equal to
\begin{equation}
C_m(t,k)=\frac{1-t/k}{1-(t/k)^m}.
\end{equation}
Since the number of files is also large, we are interested in asymptotic capacity values. So, we let $m\to \infty$, and we have that
\begin{equation}\label{ascapacity}
C(t,k)\triangleq\lim_{m\to\infty} C_m(t,k)=1-\frac{t}{k}\,.
\end{equation}

Most of the current schemes assume that servers are honest-but-curious and that they provide correct responses.
However, such an assumption cannot be guaranteed within a cloud environment. 
This poses an interesting question: what can the user do if servers provide wrong responses? 
Here, we provide three different interpretations of this question and their formal definitions. 
%Let us write down the resistance levels in increasing order and formally define them.
\begin{itemize}
\item \textit{$s$-verifiablility [referred as $s$-security in \cite{zhangwang}].} The client can detect the presence of up to $s$ servers that persuade the client to output a wrong result.
\item \textit{$a$-accountability}. The client can identify each of up to $a$ servers that persuade the client to output a wrong result.
\item \textit{$b$-byzantine resistance/$b$-byzantine robustness}. The client can retrieve the correct result in presence of up to $b$ servers that persuade the client to output a wrong result.
\end{itemize}

It is clear that $a$-accountability implies $a$-verifiablility, while $b$-byzantine resistance implies both $b$-accountability and $b$-verifiablility. 
We also note that existing verifiable PIR schemes \cite{zhang2, zhao, cao2023} provide accountability - but rely on computational assumptions and require a trusted setup. 
Other previously studied schemes are $b$-byzantine-resistant PIR schemes \cite{BPIR1, BPIR2, BPIR3, BPIR4, BPIR5}. Typically, they rely on error-correcting techniques and require at least three servers, 
while protocols considered in this paper are deployed in the two-server scenario. 

Both verifiable $a$-accountable and $b$-byzantine PIR schemes identify malicious servers. 
However, in certain low-latency applications, such as private media browsing, this may be excessive \cite{multimedia}.
By simply requiring $s$-verifiability, we can then have some savings of communication cost. 
$s$-verifiable PIR schemes were considered in papers \cite{zhang, zhangwang, cao2023}. 
In \cite{zhang, zhangwang}, authors measure the communication cost as the sum of upload and download costs, while in our case, we ignore the upload cost. In another recent paper \cite{cao2023}, authors employ linear-map commitment on the top of different PIR schemes, including PIR with optimal download rate. The fundamental difference between \cite{cao2023} and our approach is the computational assumptions required by employed commitment schemes. 
%{\color{red}Despite multi-server verifiable computation schemes from \cite{zhangwang} can be applied on top of PIR schemes with optimal download rate resulting in verifiable PIR schemes with optimal download rate, the two-server scheme cannot be generalized to publicly verifiable setup and does not offer an option to further save download cost by moving to computational guarantees of security. }

In this paper, we consider the notion of $s$-verifiability and propose a two-server verifiable PIR scheme with an optimized download rate by modifying a linear secret-sharing-based PIR scheme that achieves PIR asymptotic capacity~\eqref{ascapacity} from \cite{BitarRouayheb}. 
We also propose a generalization that allows the PIR scheme to be publicly verifiable.
%Motivated by increasing demand for publicly verifiable schemes, we propose a corresponding generalization. 
After which, to reduce the communication cost, we introduce the use of homomorphic hashing in the verification step \cite{Gkantsidis}. 

\section{Preliminaries}
\subsection{Notation}
For any integer $n>0$ we denote $[n]=\{1,\ldots,n\}$. For any prime number $p$, we denote an extended finite field with $p^t$ elements as $\mathbb{F}_{p^t}$. The base field with $p$ elements is denoted as $\mathbb{F}_{p}$. By superscript
$T$, we denote the transpose of a vector. 
\subsection{PIR Based Secret-Sharing Schemes}
In what follows, we focus on the two-server scenario \textcolor{black}{with secure communications in the system} and add a result verification to the linear secret-sharing-based PIR scheme that achieves PIR capacity~\cite[Proposition~2]{BitarRouayheb}. Let us formally define it and denote as $\Pi_0$. Each server $S_1$ and $S_2$ stores $m$ files $x_1,\ldots,x_m\in\mathbb{F}_{p^t}$ that form the data vector $\mathbf{x}^T=(x_1,\ldots,x_m)$. Queries from client to servers to retrieve the file $i$ are formed by Shamir secret sharing scheme as $\mathbf{f}(u_1)$ and $\mathbf{f}(u_2)$, where $\mathbf{f}(u)=\mathbf{e}_i+\mathbf{r}u\in\mathbb{F}_p^m$, $\mathbf{e}_i$ is all zero vector of length $m$ with a `1' in the position $i$, $\mathbf{r}$ is a random vector over $\mathbb{F}^m_p$, $u_1$ and $u_2$ are pre-selected points over $\mathbb{F}_p$ that correspond to servers one and two. As a response, server $S_j$ computes the scalar product of data vector $\mathbf{x}$ and query $\mathbf{f}(u_j)$ that can be written as $\mathbf{f}(u)\mathbf{x}^T=\mathbf{e}_i^T\mathbf{x}+\mathbf{r}^T\mathbf{x}u=x_i+\mathbf{r}^T\mathbf{x}u$. By collecting responses from two servers, the client can retrieve $x_i$ from
\begin{equation}
\begin{bmatrix}
    x_i\\ 
    \mathbf{r}^T\mathbf{x} 
    \end{bmatrix}
    =
    \begin{bmatrix}
    1 & u_1\\ 
    1 & u_2  
    \end{bmatrix}^{-1}
    \cdot
    \begin{bmatrix}
    \mathbf{f}(u_1)^T\mathbf{x}\\ 
    \mathbf{f}(u_2)^T\mathbf{x}  
    \end{bmatrix}
\end{equation}
We do note that such a scheme has the download rate $\frac{1}{2}$ and its $1$-privacy is ensured by the properties of underlined Shamir's secret sharing scheme. 

\subsection{PIR Model}
Before we proceed further, let us formally define two server verifiable PIR. Our model consist of two servers $S_1$ and $S_2$. Each server stores $m$ files $x_1,\ldots,x_m\in\mathbb{F}_{p^t}$ that form the data vector $\mathbf{x}^T=(x_1,\ldots,x_m)$. The client wants to privately retrieve the correct value of $x_i$, while one server can be malicious and provide a wrong response. 
\begin{definition}[PIR with results verification]
A two-server PIR scheme $\Pi$ with {\em results verification} consists of three algorithms that can be described as follows:
\begin{itemize}
    \item $(\textrm{vk}, \sigma_1, \sigma_2, \textrm{aux}_A, \textrm{aux}_V)\leftarrow \texttt{QueriesGen}(m, i)$ is a randomized query-generating algorithm for the client. As input, it takes the database size $m$ and retrieval index $i$ and outputs two queries $\sigma_1$ and $\sigma_2$ that will be given to servers $S_1$ and $S_2$, value $\textrm{vk}$ further employed by the client for verification purposes and, if necessary, auxiliary information $\textrm{aux}_A$ used in answer-generation and  $\textrm{aux}_V$ used in verification.  
    \item $\pi_j\leftarrow\texttt{AnswerGen}(j, \sigma_j,\textbf{x}, \textrm{aux}_A)$ is a deterministic answer-generation algorithm for server $S_j$. As input it takes server number $j$, query $\sigma_j$, database $\textbf{x}$, and, if necessary, auxiliary information $\textrm{aux}_A$, and outputs the query response $\pi_j$\,.
    \item $\{x_i,\perp\}\leftarrow\texttt{Verify}(i, \textrm{vk}, \pi_1, \pi_2, \textrm{aux}_V)$ is a deterministic verification algorithm for the client. As input, it takes the retrieval index $i$, verification key $\textrm{vk}$, servers answers $\pi_1$, $\pi_2$, if necessary, auxiliary information $\textrm{aux}_V$, and uses them to determine if it can reconstruct the correct value of $x_i$. If it cannot do so, it outputs the special symbol $\perp$ indicating that at least one of the answers is incorrect, otherwise, it reconstructs $x_i$ and outputs it.
\end{itemize}
\end{definition}

A scheme is called {\em publicly verifiable} if $\textrm{vk}$ is public. Otherwise, the scheme is {\em privately verifiable}. 
Public verification is generally preferred, but usually relies on computational assumptions \textcolor{black}{since it involves disclosing of verification key}, while privately verifiable schemes can be information-theoretically secure. Considered PIR scheme should satisfy correctness, privacy, and security properties. Let us formally define them. \textcolor{black}{Intuitively, scheme~$\Pi$ is correct if it always outputs the correct value of $x_i$ when both servers are honest.}
\begin{definition}[Correctness] The scheme $\Pi$ is {\em correct} if for any $m$, $i\in[m]$ and $\mathbf{x}\in\mathbb{F}_{p^t}^m$ and any $(\textrm{vk},\sigma_1,\sigma_2, \textrm{aux}_A, \textrm{aux}_V)\leftarrow\texttt{QueriesGen}(m,i)$ it holds that  $\texttt{Verify}(i, \textrm{vk},\texttt{AnswerGen}(1,\sigma_1,\textbf{x}, \textrm{aux}_A), \texttt{AnswerGen}(2,\sigma_2,$
$\textbf{x},\textrm{aux}_V)=x_i$.
\end{definition}
\textcolor{black}{The scheme~$\Pi$ is said to be private if each server gets no extra information about the retrieval index $i$.}
\begin{definition}[Privacy] The scheme $\Pi$ is {\em private} if for any $m$, any $i,i'\in[m]$ and $(\textrm{vk}, \sigma_1, \sigma_2, \textrm{aux}_A, \textrm{aux}_V)\leftarrow\texttt{QueriesGen}(m, i)$, $(\textrm{vk}', \sigma'_1, \sigma'_2, \textrm{aux}_A, \textrm{aux}_V)\leftarrow\texttt{QueriesGen}(m, i')$ and any $j\in[2]$, queries $\sigma_j$ and $\sigma'_j$ are identically distributed. The latter means that a given server $j$ cannot distinguish between them.      
\end{definition}
\textcolor{black}{Intuitively, the scheme~$\Pi$ is verifiable if no single server can cause the client to output a wrong value $x_i$ by providing undetected wrong answer.} Following \cite{VC1} and recent papers on multi-server verifiable computations \cite{zhangwang, zhang}, 
the verifiablity property can be defined through the notion of security experiment \textcolor{black}{between more powerful adversary and challenger acting as a client}: 
an adversary $\mathcal{A}$ controls the dishonest server $S_j$, 
knows the database $\mathbf{x}$, retrieval index $i$, and, if necessary, the auxiliary information $\textrm{aux}_A$ and crafts an answer $\hat{\pi}_j$ after receiving the query $\sigma_j$. 
We consider this setup for the privately verifiable case and we denote the corresponding experiment as $\textrm{EXP}_{\mathcal{A},\Pi}^{PriV}(m,\mathbf{x},i,j)$. 
For the publicly verifiable case, we borrow the same ideas from \cite{zhangwang} and 
generalize the security experiment of \cite{VC2} in a single-server setting to our two-server setup. 
The resulting security experiment $\textrm{EXP}_{\mathcal{A},\Pi}^{PubV}(m,\mathbf{x},i,j)$ is identical to $\textrm{EXP}_{\mathcal{A},\Pi}^{PriV}(m,\mathbf{x},i,j)$ except that adversary also knows verification key $\textrm{vk}$. 
Let us formally define them.
\begin{definition}[Security experiment]
An {\em interactive security experiment}  $\textrm{EXP}_{\mathcal{A},\Pi}^{PriV}(m,\mathbf{x},i,j)/\textrm{EXP}_{\mathcal{A},\Pi}^{PubV}(m,\mathbf{x},i,j)$ between adversary $\mathcal{A}$ that controls dishonest server $j\in[2]$ and its challenger in privately/publicly verifiable case can be described as follows:
\begin{itemize}
    \item Challenger generates $(\textrm{vk}, \sigma_1, \sigma_2, \textrm{aux}_A, \textrm{aux}_V)\leftarrow\texttt{QueriesGen}(m, i)$ and sends $\sigma_j$ to the adversary $\mathcal{A}$.
    \item The adversary $\mathcal{A}$ generates a crafted answer $\hat{\pi}_j\leftarrow\mathcal{A}(j, \sigma_j, \mathbf{x},i, \textrm{aux}_A)/\hat{\pi}_j\leftarrow\mathcal{A}(j, \sigma_j, \mathbf{x},i, \textrm{vk},\textrm{aux}_A)$ and sends it to the challenger.
    \item The challenger computes $\pi_{3-j}\leftarrow\texttt{AnswerGen}(3-j,\sigma_{3-j},\textbf{x},\textrm{aux}_A)$
    \item The challenger runs the verification algorithm $\texttt{Verify}$ with inputs $i$, $\textrm{vk}$, $\hat{\pi}_j$ and $\pi_{3-j}$, if necessary, $\textrm{aux}_V$ and computes an output $y$.
    \item If $y\notin\{x_i,\perp\}$ set the outcome $\textrm{EXP}_{\mathcal{A},\Pi}^{PriV}(m,\mathbf{x},i,j)/\textrm{EXP}_{\mathcal{A},\Pi}^{PubV}(m,\mathbf{x},i,j)$ of the experiment to be $1$, otherwise set it to be $0$.
\end{itemize}
\end{definition}
In what follows, we define the following two notions of verifiablity -- information-theoretical and computational. 

\begin{definition}[Negligible function]
A function from $\mathbb{N}$ to $\mathbb{R}^{+}$ is {\em negligible} and denoted as $\texttt{negl}(.)$ if for all $c>0$ there exists a natural number $\lambda_0$ such that $\texttt{negl}(\lambda)<\frac{1}{\lambda^c}$ for all $\lambda>\lambda_0$.
\end{definition}

\begin{definition}[Information-theoretical virifiability]
The protocol $\Pi$ is {\em $(1,\epsilon)$-verifiable} if for any adversary $\mathcal{A}$, any $j\in[2]$, $m$, $\mathbf{x}\in\mathbb{F}_{p^t}^m$ and any $i\in[m]$, we have $\textrm{Pr}[\textrm{EXP}_{\mathcal{A},\Pi}^{PriV}(m,\mathbf{x},i,j)=1]\leq\epsilon$. 
Here, the probability is taken over the randomness of $\mathcal{A}$ and the experiment. 
\end{definition} 

\begin{definition}[Computational verifiability]
The scheme $\Pi$ is $1$-{\em verifiable} if for any probabilistic poly-time (PPT) adversary $\mathcal{A}$, any $j\in[2]$, $m$, $\mathbf{x}\in\mathbb{F}_{q^t}^m$ and any $i\in[m]$, \textcolor{black}{$\textrm{Pr}[\textrm{EXP}_{\mathcal{A},\Pi}^{PriV}(m,\mathbf{x},i,j)]\leq \texttt{negl}(\lambda)/\textrm{Pr}[\textrm{EXP}_{\mathcal{A},\Pi}^{PubV}(m,\mathbf{x},i,j)=1]\leq \texttt{negl}(\lambda)$}, where the probability is taken over the randomness of $\mathcal{A}$ and the experiment. 
\end{definition}

The notion of computational verifiability relies on certain cryptographic assumptions. 
Schemes proposed in this paper rely on the following Discrete logarithm (DLog) assumption. %Let us formally define it. 

\begin{definition}[DLog assumption]
Let $\mathbb{G}$ be a cyclic group of order $p>2^{\lambda}$. For generator $g$ and $\alpha\in\mathbb{F}_{p}\setminus\{0\}$ we define the advantage $\textrm{Adv}_{\mathcal{A}}^{DLog}$ of adversary $\mathcal{A}$ in solving discrete logarithm problem in group $\mathbb{G}$ as the probability that he can find $\alpha$ from values of $g$ and $g^{\alpha}$. We say that discrete logarithm assumption holds if for any PPT adversary $\mathcal{A}$, $\textrm{Adv}_{\mathcal{A}}^{DLog}(\lambda)\leq \texttt{negl}(\lambda)$. 
\end{definition}
\section{Our Constructions}
In this section, we introduce two-server PIR schemes in which the client can verify the result in the presence of one dishonest server. 
The basic idea is to provide one more independent set of queries to the servers, so that the user can compute two versions of $x_i$ or $x_i$ and $\textrm{v}x_i$ for some secret parameter $\textrm{v}$ and verify that both servers are truthful.
If we keep the verification key in secret and download whole responses to queries, we obtain an information-theoretical $(1,\epsilon)$-verifiable scheme with private verification, or shortly $(1,\epsilon)$ privately verifiable scheme (see section~\ref{ITVPIR}). Later on, we generalize this scheme to a publicly verifiable setup (see section~\ref{PUBVER}) and reduce the communication cost by introducing homomorphic hashes (see section~\ref{Hashes}).  

\subsection{Information-Theoretical Privately Verifiable Scheme}\label{ITVPIR}

Let us modify the linear secret-sharing-based PIR scheme $\Pi_0$ to the case of results verification by creating one more set of independent queries. The resulting two server verifiable PIR scheme $\Pi_1$ can be described as follows.
\begin{itemize}
    \item $\texttt{QueriesGen}(m,i)$: Choose $\textrm{v}\leftarrow \mathbb{F}_p\setminus\{0\}$, $\textbf{r}\leftarrow \mathbb{F}_p^m$, $\textbf{r}_v\leftarrow \mathbb{F}_p^m$. Let $\mathbf{f}(u)=\mathbf{e}_i+\mathbf{r}u\in\mathbb{F}_p^m$ and $\mathbf{f}_v(u)=\textrm{v}\mathbf{e}_i+\mathbf{r}_vu\in\mathbb{F}_p^m$, $\mathbf{e}_i$ is all zero vector of length $m$ with a `1' in the position $i$. Compute $\sigma_j=(\mathbf{f}(u_j), \mathbf{f}_v(u_j))$ for all $j\in[2]$. Output $\textrm{vk}=\textrm{v}$, $\sigma_1$, $\sigma_2$, and $\textrm{aux}_V=\{u_1, u_2\}$. 
    \item $\texttt{AnswerGen}(j, \sigma_j, \mathbf{x})$: Parse $\sigma_j$ as $(\mathbf{f}(u_j), \mathbf{f}_v(u_j))$, compute $z_j=\mathbf{f}(u_j)^T\mathbf{x}$, $w_j=\mathbf{f}_v(u_j)^T\mathbf{x}$, output $\pi_j=(z_j, w_j)$.
    \item $\texttt{Verify}(i,\textrm{vk},\pi_1,\pi_2, \textrm{aux}_V)$: Parse $\textrm{aux}_V$ as $\{u_1, u_2\}$. Retrieve $x_i$ and $\textrm{v}x_i$ as 
    \begin{equation}\label{1}
\begin{bmatrix}
    x_i\\ 
    \mathbf{r}^T\mathbf{x} 
    \end{bmatrix}
    =
    \begin{bmatrix}
    1 & u_1\\ 
    1 & u_2  
    \end{bmatrix}^{-1}
    \cdot
    \begin{bmatrix}
    z_1\\ 
    z_2  
    \end{bmatrix}
        =
    \begin{bmatrix}
    a & b\\ 
    c & d  
    \end{bmatrix}
    \cdot
    \begin{bmatrix}
    z_1\\ 
    z_2  
    \end{bmatrix}
\end{equation}
and 
\begin{equation}\label{2}
\begin{bmatrix}
    \textrm{v}x_i\\ 
    \mathbf{r}_v^T\mathbf{x} 
    \end{bmatrix}
    =
    \begin{bmatrix}
    1 & u_1\\ 
    1 & u_2  
    \end{bmatrix}^{-1}
    \cdot
    \begin{bmatrix}
    w_1\\ 
    w_2 
    \end{bmatrix}
        =
    \begin{bmatrix}
    a & b\\ 
    c & d  
    \end{bmatrix}
    \cdot
    \begin{bmatrix}
    w_1\\ 
    w_2  
    \end{bmatrix}
\end{equation}

If the equation below holds, output $x_i$; otherwise, output $\perp$.
\begin{equation}
 \textrm{v}\left(az_1+bz_2\right)=aw_1+bw_2\end{equation}
\end{itemize}
The proofs of correctness of our scheme and proof of privacy are almost identical to that of~\cite{BitarRouayheb} and omitted here. The download rate is equal to $\frac{1}{4}$ and below we show that $\Pi_1$ is $(1,\epsilon)$-verifiable.

\begin{theorem}

The scheme $\Pi_1$ is $(1,\epsilon)$-verifiable where $\epsilon=\frac{1}{p-1}$.
\end{theorem}
\begin{proof}
Without loss of generality, assume that Adversary $\mathcal{A}$ controls the server $S_1$ and set $j=1$. Let $\pi_1=(z_1, w_1)$ and $\pi_2=(z_2, w_2)$ be the answers obtained by correctly executing algorithm $\texttt{AnswerGen}$ by each server. Let $\hat{\pi}_1=(\hat{z}_{1}, \hat{w}_{1})$ be the answer chosen by $\mathcal{A}$ for $S_1$.

From the description of $\Pi_1$ it is clear that 
$$x_i=A=az_1+bz_2$$ and $$\textrm{v}x_i=V=aw_1+bw_2$$ while 
$$\hat{x}_i=\hat{A}=a\hat{z}_{1}+bz_2$$ and
$$\hat{\textrm{v}}\hat{x}_i=\hat{V}=a\hat{w}_{1}+bw_2.$$

$\mathcal{A}$ wins the security experiment $\textrm{EXP}_{\mathcal{A},\Pi_1}^{PriV}(m,\mathbf{x},i,j)$ if the $\hat{A}\ne A$ and $\hat{V}=\textrm{v}\hat{A}$. It is clear that $V=\textrm{v}A$.
Hence, in this case, $$\hat{V}-V=\textrm{v}(\hat{A}-A).$$
From equations~\eqref{1}, \eqref{2} and the fact that server $S_2$ is honest, it is clear that
\begin{equation}
\hat{A}-A=a\left(\hat{z}_1-z_1\right)=a\Delta_0\ne0
\end{equation}
\begin{equation}
\hat{V}-V=a\left(\hat{w}_1-w_1\right)=a\Delta_1
\end{equation}
Hence $\mathcal{A}$  wins the security  experiment $\textrm{EXP}_{\mathcal{A},\Pi_1}^{PriV}(m,\mathbf{x},i,j)$ iff it finds the solution for the equation
\begin{equation}
G(\textrm{v})=a\Delta_1-a\textrm{v}\Delta_0=a(\Delta_1-\textrm{v}\Delta_0)=0.
\end{equation}
From the description of the security experiment, it follows that $\Delta_0, \Delta_1$ are known to $\mathcal{A}$ and independent from $\textrm{v}$. Hence 
\begin{equation}
    \textrm{Pr}[\textrm{EXP}_{\mathcal{A},\Pi_1}^{PriV}(m,\mathbf{x},i,j)=1]\leq\textrm{Pr}[G(\textrm{v})=0],
\end{equation}
for $\textrm{v}$ chosen independently and uniformly at random from $\mathbb{F}_p\setminus\{0\}$. \textcolor{black}{Since $\Delta_0\ne0$,} by Schwartz-Zippel Lemma \cite{SZ1, SZ2} this value can be estimated from above as $\frac{1}{p-1}$, and the theorem statement follows.
\end{proof}

\subsection{Computational Publicly Verifiable Scheme}\label{PUBVER}

Next, we modify information-theoretical privately verifiable scheme $\Pi_1$ to a publicly verifiable setup. 
To do so, we choose a cyclic group $\mathbb{G}$ with generator $g$ of prime order $p\geq 2^{\lambda}$ and made $\textrm{vk}=g^{\textrm{v}}$ public. The resulting PIR scheme $\Pi_2$ can be described as follows.

\begin{itemize}
    \item $\texttt{QueriesGen}(m,i)$: Choose $\textcolor{black}{\textrm{v}}\leftarrow \mathbb{F}_p\setminus\{0\}$, $\textbf{r}\leftarrow \mathbb{F}_p^m$, $\textbf{r}_v\leftarrow \mathbb{F}_p^m$. Let $\mathbf{f}(u)=\mathbf{e}_i+\mathbf{r}u\in\mathbb{F}_p^m$ and $\mathbf{f}_v(u)=\textrm{v}\mathbf{e}_i+\mathbf{r}_vu\in\mathbb{F}_p^m$. 
    %$\mathbf{e}_i$ is all zero vector of length $m$ with a `1' in the position $i$. 
    Compute $\sigma_j=(\mathbf{f}(u_j), \mathbf{f}_v(u_j))$ for all $j\in[2]$. Choose a cyclic group $\mathbb{G}$ of prime order $p$ with generator $g$.  Output $\textrm{vk}=g^{\textrm{v}}$, $\sigma_1$, $\sigma_2$ and $\textrm{aux}_V=\{u_1, u_2\}$.
    \item $\texttt{AnswerGen}(j, \sigma_j, \mathbf{x})$: Parse $\sigma_j$ as $(\mathbf{f}(u_j), \mathbf{f}_v(u_j))$, compute $z_j=\mathbf{f}(u_j)^T\mathbf{x}$, $w_j=\mathbf{f}_v(u_j)^T\mathbf{x}$, output $\pi_j=(z_j, w_j)$.
    \item $\texttt{Verify}(i,\textrm{vk},\pi_1,\pi_2, \textrm{aux}_V)$: Parse $\textrm{aux}_V$ as $\{u_1, u_2\}$. Retrieve $x_i$ and $\textrm{v}x_i$ as 
    \begin{equation}\label{3}
\begin{bmatrix}
    x_i\\ 
    \mathbf{r}^T\mathbf{x} 
    \end{bmatrix}
    =
    \begin{bmatrix}
    1 & u_1\\ 
    1 & u_2  
    \end{bmatrix}^{-1}
    \cdot
    \begin{bmatrix}
    z_1\\ 
    z_2  
    \end{bmatrix}
        =
    \begin{bmatrix}
    a & b\\ 
    c & d  
    \end{bmatrix}
    \cdot
    \begin{bmatrix}
    z_1\\ 
    z_2 
    \end{bmatrix}
\end{equation}
and 
\begin{equation}\label{4}
\begin{bmatrix}
    \textrm{v}x_i\\ 
    \mathbf{r}_v^T\mathbf{x} 
    \end{bmatrix}
    =
    \begin{bmatrix}
    1 & u_1\\ 
    1 & u_2  
    \end{bmatrix}^{-1}
    \cdot
    \begin{bmatrix}
    w_1\\ 
    w_2 
    \end{bmatrix}
        =
    \begin{bmatrix}
    a & b\\ 
    c & d  
    \end{bmatrix}
    \cdot
    \begin{bmatrix}
    w_1\\ 
    w_2 
    \end{bmatrix}
\end{equation}

If the equation below holds, output $x_i$; otherwise, output $\perp$:
\begin{equation}
 (g^{\textrm{v}})^{\left(az_1+bz_2\right)}=g^{aw_1+bw_2}.\end{equation}
\end{itemize}
The proofs of correctness of our scheme and proof of privacy are almost identical to that of  \cite{BitarRouayheb} and omitted here. The download rate is equal to $\frac{1}{4}$. 

\begin{theorem}
The scheme $\Pi_2$ is $1$-verifiable under DLog assumption in $\mathbb{G}$.
\end{theorem}
\begin{proof}
Without loss of generality, assume that Adversary $\mathcal{A}$ controls the server $S_1$ and set $j=1$. Let $\pi_1=(\mathbf{f}(u_1)^T\mathbf{x},\mathbf{f}_v(u_1)^T\mathbf{x})=(z_1, w_1)$ and $\pi_2=(\mathbf{f}(u_2)^T\mathbf{x},\mathbf{f}_v(u_2)^T\mathbf{x})=(z_2, w_2)$ be the answers obtained by correctly executing algorithm $\texttt{AnswerGen}$ by each server. Let $\hat{\pi}_1=(\hat{z}_1, \hat{w}_1)$ is answers chosen by $\mathcal{A}$ for $S_1$. 

From the description of $\Pi_1$ it is clear that $$x_i=A=az_1+bz_2$$ and $$\textrm{v}x_i=V=aw_1+bw_2$$ while $$\hat{x}_i=\hat{A}=a\hat{z}_1+bz_2$$ and
$$\hat{\textrm{v}}\hat{x}_i=\hat{V}=a\hat{w}_1+bw_2.$$

$\mathcal{A}$ wins the security experiment $\textrm{EXP}_{\mathcal{A},\Pi_2}^{PubV}(m,\mathbf{x},i,j)$ if the $\hat{A}\ne A$ and $g^{\hat{V}}=g^{\textrm{v}\hat{A}}$.  From equations~\eqref{3},\eqref{4} and the fact that server $2$ is honest it is clear that
\begin{equation}
\hat{A}-A=a\left(\hat{z}_1-z_1\right)=a\Delta_0\ne0
\end{equation}
\begin{equation}
\hat{V}-V=a\left(\hat{w}_1-w_1\right)=a\Delta_1.
\end{equation}
As $g^{V}=(g^{\textrm{v}})^{A}$, $\mathcal{A}$ wins the security experiment $\textrm{EXP}_{\mathcal{A},\Pi_2}^{PubV}(m,\mathbf{x},i,j)$ iff \textcolor{black}{$(g^{\textrm{v}})^{a\Delta_0}=g^{a\Delta_1}$.} From the description of security experiment it follows that $\Delta_0$, $\Delta_1$ are known to $\mathcal{A}$ and independent from $\textrm{v}$ and $g^{\textrm{v}}$. Hence  $\textrm{Pr}[\textrm{EXP}_{\mathcal{A},\Pi_1}^{PriV}(m,\mathbf{x},i,j)=1]$ can be estimated from above as probability of learning the value $\textrm{v}=\frac{a\Delta_0}{a\Delta_1}=\frac{\Delta_0}{\Delta_1}$ where $\Delta_0\ne0$ from the discrete-logarithm relationship in the group $\mathbb{G}$. As a result, the theorem statement follows.
\end{proof}

\begin{table*}[ht!]
\centering
\scalebox{1}{
\begin{tabular}{|l|l|l|l|l|l|l|l|}
\hline
                     & $\Pi_0$ & \cite[Scheme~3]{zhangwang} & $\mathbb{A}$ & $\Pi_1$      & $\Pi_2$          & $\Pi_3$                         \\ \hline
File size            & $t\cdot\log(p)$                                                      & $t\cdot\log(p)$                                                & $t\cdot\log(p)$                                                         & $t\cdot\log(p)$  & $t\cdot\log(p)$       & $t\cdot\log(p)$                  \\ \hline
UC               & $2m\cdot\log(p)$                                                     & $4m\cdot\log(p)$                                               & $4m\cdot\log(p)$                                                        & $4m\cdot\log(p)$ & $4m\cdot\log(p)$      & $4m\cdot\log(p)$                    \\ \hline
DC             & $2t\cdot\log(p)$                                                     & $4t\cdot\log(p)$                                               & $4t\cdot\log(p)$                                                        & $4t\cdot\log(p)$ & $4t\cdot\log(p)$     & $2t\cdot\log(p)+2\cdot\log(r)$    \\ \hline
Pr& $1$                                                            & $\frac{p-1}{p^2-3}$                                         & $\frac{2(p-1)}{(p-2)^2 }$                                                 & $\frac{1}{p-1}$  & $\texttt{negl}(\lambda)$ & $\texttt{negl}(\lambda)$          \\ \hline
1-verif.           & no                                                           & IT                                                     & IT                                                              & IT       & DLog   & DLog      \\ \hline
1-priv.            & IT                                                           & IT                                                     & IT                                                              & IT       & IT            & IT                         \\ \hline
verif.        & no                                                           & private                                                & private                                                         & private  & public        & private                    \\ \hline
\end{tabular}
}
\caption{Comparison of two-server PIR schemes that optimize the download rate}\label{comp}

\end{table*}

\subsection{Computational Privately Verifiable Scheme}\label{Hashes}

The main drawback of scheme $\Pi_1$ is that it doubles the download cost in comparison to capacity-achieving scheme $\Pi_0$. The possible solution to a problem of this kind is to apply homomorphic hashes to the second part of responses $\pi_1$ and $\pi_2$. The construction of homomorphic hashes is based on DLog assumption in the multiplicative group of order $p\geq 2^{\lambda}$ in the finite field and was introduced for the first time for verification of digital content distributed by rateless erasure codes in \cite{Krohn} and further applied for network coding in \cite{Gkantsidis}. Resulted PIR scheme $\Pi_3$ can be described as follows.

\begin{itemize}
    \item $\texttt{QueriesGen}(m,i)$: Choose $\textcolor{black}{\textrm{v}}\leftarrow \mathbb{F}_p\setminus\{0\}$, $\textbf{r}\leftarrow \mathbb{F}_p^m$, $\textbf{r}_v\leftarrow \mathbb{F}_p^m$. Let $\mathbf{f}(u)=\mathbf{e}_i+\mathbf{r}u\in\mathbb{F}_p^m$ and $\mathbf{f}_v(u)=\textrm{v}\mathbf{e}_i+\mathbf{r}_vu\in\mathbb{F}_p^m$.
    % $\mathbf{e}_i$ is all zero vector of length $m$ with a `1' in the position $i$. 
    Compute $\sigma_j=(\mathbf{f}(u_j), \mathbf{f}_v(u_j))$ for all $j\in[2]$.
    Choose a prime number $r$ so that $r-1$ is divisible by $p$ and random elements $g_1,\ldots,g_t$ of $\mathbb{Z}_r$  of order $p$. Choose a basis $\mathcal{B}$ of $\mathbb{F}_{p^t}$ over $\mathbb{F}_p$. Output $\textrm{vk}=\textrm{v}$, $\sigma_1$, $\sigma_2$, $\textrm{aux}_A=\{g_1,\ldots,g_t,\mathcal{B}\}$ and $\textrm{aux}_V=\{g_1,\ldots,g_t,\mathcal{B}, u_1, u_2\}$.
    \item $\texttt{AnswerGen}(j, \sigma_j, \mathbf{x}, \textrm{aux}_A)$: Parse $\sigma_j$ as $(\mathbf{f}(u_j), \mathbf{f}_v(u_j))$, $\textrm{aux}_A$ as $\{g_1,\ldots,g_t,\mathcal{B}\}$,  compute $z_j=\mathbf{f}(u_j)^T\mathbf{x}$, $w_j=\mathbf{f}_v(u_j)^T\mathbf{x}$. Represent $w_j$ in basis $\mathcal{B}$ of $\mathbb{F}_{p^t}$ over $\mathbb{F}_p$ as $(w_{j,1},\ldots,w_{j,t})^T$ and compute $h(w_j)=\prod_{l=1}^tg_l^{w_{j,l}}\;\;\textrm{mod}\;r$. Output $\pi_j=(z_j, h(w_j))$.
    \item $\texttt{Verify}(i,\textrm{vk},\pi_1, \pi_2, \textrm{aux}_V)$: Parse $\textrm{aux}_V$ as $\{g_1,\ldots,g_t,\mathcal{B}, u_1, u_2\}$.  Retrieve $x_i$ as 
    \begin{equation}\label{5}
\begin{bmatrix}
    x_i\\ 
    \mathbf{r}^T\mathbf{x} 
    \end{bmatrix}
    =
    \begin{bmatrix}
    1 & u_1\\ 
    1 & u_2  
    \end{bmatrix}^{-1}
    \cdot
    \begin{bmatrix}
    z_1\\ 
    z_2  
    \end{bmatrix}
        =
    \begin{bmatrix}
    a & b\\ 
    c & d  
    \end{bmatrix}
    \cdot
    \begin{bmatrix}
    z_1\\ 
    z_2  
    \end{bmatrix}
\end{equation}

Represent $x_i$ in basis $\mathcal{B}$ of $\mathbb{F}_{p^t}$ over $\mathbb{F}_{p}$ as $(x_{i,1},\ldots,x_{i,t})^T$ and compute $h(x_i)=\prod_{l=1}^tg_l^{x_{i,l}}\;\;\textrm{mod}\;r$. 

If the equation below holds, output $x_i$; otherwise, output $\perp$.
\begin{equation}
h(x_i)^{\textrm{v}}=h(w_1)^ah(w_2)^b\;\;\textrm{mod}\;r.
\end{equation}
\end{itemize}

The proofs of correctness of our scheme and proof of privacy are almost identical to that of  \cite{BitarRouayheb} and omitted here. The download rate is equal to $\frac{t\log p}{2(t\log p +\log r)}$ that can be arbitrarily close to $\frac{1}{2}$.

\begin{theorem}
The scheme $\Pi_3$ is $1$-verifiable under DLog assumption in a multiplicative group of order $p$ in a finite field. 
\end{theorem}

\begin{proof}
Without loss of generality, assume that Adversary $\mathcal{A}$ controls the server $S_1$ and set $j=1$. Let $\pi_1=(\mathbf{f}(u_1)^T\mathbf{x},h(\mathbf{f}_v(u_1)^T\mathbf{x}))=(z_1, h(w_1))$ and $\pi_2=(z_2, h(w_2))$ be the answers obtained by correctly executing algorithm $\texttt{AnswerGen}$ by each server. Let $\hat{\pi}_1=(\hat{z}_1, h(\hat{w}_1)$ be the answer chosen by $\mathcal{A}$ for $S_1$. 

Let $HC$ and $HC^c$ denote the event of hash collision and its complement, respectively. 
From the description of $\Pi_3$ and the law of total probability,
it is clear that 
\begin{comment}
\begin{equation*}
	\textrm{Pr}[\textrm{EXP}_{\mathcal{A},\Pi_3}^{PriV}(m,\mathbf{x},i,j)=1]=
	\textrm{Pr}[\textrm{EXP}_{\mathcal{A},\Pi_3}^{PriV}(m,\mathbf{x},i,j)=1|
\end{equation*}
\begin{equation*}
	\textrm{collision}\;\textrm{in}\;\textrm{hash}]\cdot\textrm{Pr}[\textrm{collision}\;\textrm{in}\;\textrm{hash}]+\textrm{Pr}[\textrm{EXP}_{\mathcal{A},\Pi_3}^{PriV}(m,\mathbf{x},i,j)=1| \end{equation*}
\begin{equation*}
	\textrm{no}\;\textrm{collision}\;\textrm{in}\;\textrm{hash}]\cdot \textrm{Pr}[\textrm{no}\;\textrm{collision}\;\textrm{in}\;\textrm{hash}]\leq\textrm{Pr}[\textrm{collision}\;\textrm{in}\;\textrm{hash}]
\end{equation*}
\begin{equation*}
	+\textrm{Pr}[\textrm{EXP}_{\mathcal{A},\Pi_3}^{PriV}(m,\mathbf{x},i,j)=1|\textrm{no}\;\textrm{collision}\;\textrm{in}\;\textrm{hash}]\cdot\textrm{Pr}[\textrm{no}\;\textrm{collision}\;\textrm{in}
\end{equation*}
\begin{equation}
	\textrm{hash}]\leq\textrm{Pr}[\textrm{collision}\;\textrm{in}\;\textrm{hash}]+\textrm{Pr}[\textrm{EXP}_{\mathcal{A},\Pi_1}^{PriV}(m,\mathbf{x},i,j)=1]\leq \texttt{negl}(\lambda).   
\end{equation}
\end{comment}
\begin{align}
&\textrm{Pr}[\textrm{EXP}_{\mathcal{A},\Pi_3}^{PriV}(m,\mathbf{x},i,j)=1] \notag\\
& = \textrm{Pr}[\textrm{EXP}_{\mathcal{A},\Pi_3}^{PriV}(m,\mathbf{x},i,j)=1|HC]\textrm{Pr}[HC] \notag\\
& ~+ \textrm{Pr}[\textrm{EXP}_{\mathcal{A},\Pi_3}^{PriV}(m,\mathbf{x},i,j)=1|HC^c]\textrm{Pr}[HC^c]\notag\\
& \leq \textrm{Pr}[HC] + \textrm{Pr}[\textrm{EXP}_{\mathcal{A},\Pi_3}^{PriV}(m,\mathbf{x},i,j)=1|HC^c]\textrm{Pr}[HC^c]\notag\\
& \leq \textrm{Pr}[HC] + \textrm{Pr}[\textrm{EXP}_{\mathcal{A},\Pi_1}^{PriV}(m,\mathbf{x},i,j)=1]\leq \texttt{negl}(\lambda).\label{last}
\end{align}
The final inequality~\eqref{last} follows from the collision resistance of homomorphic hash functions under the DLog assumption in a multiplicative group of order $p\geq 2^{\lambda}$.
\end{proof}
\begin{remark}
%We can replace homomorphic hashes construction based on DLog assumption with homomorphic hashes construction based on t-poly DH assumption in the group of points of an elliptic curve \cite{DH1, DH2}. Despite the latter requiring a smaller field size to deploy, it requires the existence of a trusted setup that can be impossible in some scenarios. 

We do note that the required properties of homomorphic hashes can be obtained from a more general cryptographic primitive known as linearly homomorphic vector commitment. More precisely, we can consider each element $y\in\FF_{p^l}$ as a vector in $\FF_p$ in the fixed basis $\mathcal{B}$ and apply the linearly homomorphic vector commitment to it in the same way as we applied homomorphic hashes before. For our purposes, it is important that the commitment scheme is binding, meaning that the commitment cannot be opened to values that are inconsistent with the committed vector under certain cryptographic assumptions. Additionally, it must have a homomorphic property, namely for any $y_1, y_2\in\FF_{p^t}$ and $\alpha_1, \alpha_2$, it holds that
\begin{align}
&\textrm{Commitment}(\alpha_1y_1+\alpha_2y_2)=\notag\\
&\alpha_1\textrm{Commitment}(y_1)+\alpha_2\textrm{Commitment}(y_2).  
\end{align}

For a detailed introduction to linearly homomorphic vector commitment schemes, we refer the reader to~\cite{Nazirkhanova2021, Kate2010}, and the references therein.
\end{remark}

\section{Comparisons}

Let us consider setup when we have $m$ files over a finite field $\mathbb{F}_{p^t}$ and possible solutions to the problem of verifiable PIR with optimized download cost. By Pr we denote the probability that an adversary with access to one server wins in the security experiment. We measure the file size, upload cost (UP), and download cost (DC) in bits. Also, we differentiate schemes by the notion of 1-verifiablity (1-verif.), 1-privacy (1-priv.), and presence of verifiability (verif.). The comparison is presented in Table~\ref{comp}. 

For comparison we take schemes $\Pi_0$, $\Pi_1$, $\Pi_2$ and $\Pi_3$ decribed in this paper. Also, we apply a two-server verifiable computation scheme of low-degree polynomials from \cite[Scheme~3]{zhangwang} at the top of scheme $\Pi_0$ to construct the two-server verifiable scheme. We note that this is the only scheme from \cite{zhangwang} that can be deployed in two server scenarios. It cannot be generalized to publicly verifiable cases, and we cannot apply homomorphic hashing schemes to it as we cannot separate server responses into two parts (one is responsible for the retrieval, and one is for the verification). The parameters of the resulting scheme are presented in column~$3$. 

Two server-verifiable schemes can also be created from scheme $\Pi_0$ by forming two independent queries for each server. The important thing here is that in comparison to previously described schemes, points of evaluation must be kept secret. Such an approach offers a higher probability that an adversary with access to one server wins in the security experiment. Let us formally describe it below and prove its $(1,\epsilon)$-verifiability. In that follows, we denote this scheme as $\mathbb{A}$. 

\begin{itemize}
    \item $\texttt{QueriesGen}(m,i)$: Choose $u_1, u_2\leftarrow \mathbb{F}_p$, $u_1\ne u_2$, $\tilde{u}_1, \tilde{u}_2\leftarrow \mathbb{F}_p$, $\tilde{u}_1\ne \tilde{u}_2$, $\textbf{r}\leftarrow \mathbb{F}_p^m$, $\textbf{r}_v\leftarrow \mathbb{F}_p^m$. Let $\mathbf{f}(u)=\mathbf{e}_i+\mathbf{r}u\in\mathbb{F}_p^m$ and $\mathbf{f}_v(u)=\mathbf{e}_i+\mathbf{r}_vu\in\mathbb{F}_p^m$.
    %, $\mathbf{e}_i$ is all zero vector of length $m$ with a `1' in the position $i$. 
    Compute $\sigma_j=(\mathbf{f}(u_j), \mathbf{f}_v(\tilde{u}_j))$ for all $j\in[2]$. Output $\textrm{vk}=\textrm{aux}_V=\{u_1, u_2, \tilde{u}_1, \tilde{u}_2\}$, $\sigma_1$ and $\sigma_2$.
    \item $\texttt{AnswerGen}(j, \sigma_j, \mathbf{x})$: Parse $\sigma_j$ as $(\mathbf{f}(u_j), \mathbf{f}_v(u_j))$, compute $z_j=\mathbf{f}(u_j)^T\mathbf{x}$, $w_j=\mathbf{f}_v(u_j)^T\mathbf{x}$, output $\pi_j=(z_j, w_j)$.
    \item $\texttt{Verify}(i,\textrm{vk},\pi_1,\pi_2, \textrm{aux}_V)$: Parse $\textrm{vk}=\textrm{aux}_V=\{u_1, u_2, \tilde{u}_1, \tilde{u}_2\}$. Retrieve $x_i$ as 
    \begin{equation}\label{12}
\begin{bmatrix}
    x_i\\ 
    \mathbf{r}^T\mathbf{x} 
    \end{bmatrix}
    =
    \begin{bmatrix}
    1 & u_1\\ 
    1 & u_2  
    \end{bmatrix}^{-1}
    \cdot
    \begin{bmatrix}
    z_1\\ 
    z_2 
    \end{bmatrix}
        =
    \begin{bmatrix}
    a & b\\ 
    c & d  
    \end{bmatrix}
    \cdot
    \begin{bmatrix}
    z_1\\ 
    z_2  
    \end{bmatrix}
\end{equation}
and 
\begin{equation}\label{22}
\begin{bmatrix}
    x_i\\ 
    \mathbf{r}_v^T\mathbf{x} 
    \end{bmatrix}
    =
    \begin{bmatrix}
    1 & \tilde{u}_1\\ 
    1 & \tilde{u}_2  
    \end{bmatrix}^{-1}
    \cdot
    \begin{bmatrix}
    w_1\\ 
    w_2  
    \end{bmatrix}
        =
    \begin{bmatrix}
    \tilde{a} & \tilde{b}\\ 
    \tilde{c} & \tilde{d}  
    \end{bmatrix}
    \cdot
    \begin{bmatrix}
    w_1\\ 
    w_2  
    \end{bmatrix}
\end{equation}

If the equation below holds, output $x_i$; otherwise, output $\perp$.
\begin{equation}
 az_1+bz_2=\tilde{a}w_1+\tilde{b}w_2\end{equation}
\end{itemize}
The proofs of correctness of this scheme and proof of privacy are almost identical to that of  \cite{BitarRouayheb} and omitted here. The download rate is equal to $\frac{1}{4}$.
\begin{theorem}\label{th::alternative}
The scheme $\mathbb{A}$ is $(1,\epsilon)$-verifiable where $\epsilon=\frac{2(p-1)}{(p-2)^2}$.
\end{theorem}
\begin{proof}
Without loss of generality, assume that Adversary $\mathcal{A}$ controls the server $S_1$ and set $j=1$. Let $\pi_1=(z_1, w_1)$ and $\pi_2=(z_2, w_2)$ be the answers obtained by correctly executing algorithm $\texttt{AnswerGen}$ by each server. Let $\hat{\pi}_1=(\hat{z}_1, \hat{w}_1)$ be the answer chosen by $\mathcal{A}$ for $S_1$.

From the description of $\mathbb{A}$ it is clear that $$x_i=A=az_1+bz_2$$ and $$x_i=V=\tilde{a}w_1+\tilde{b}w_2$$ while $$\hat{x}_i=\hat{A}=a\cdot\hat{z}_1+bz_2$$ and
$$\tilde{x}_i=\hat{V}=\tilde{a}\hat{w}_1+\tilde{b}w_2.$$

$\mathcal{A}$ wins the security experiment $\textrm{EXP}_{\mathcal{A},\mathbb{A}}^{PriV}(m,\mathbf{x},i,j)$ if the $\hat{A}\ne A$ and $\hat{V}=\hat{A}$. It is clear that $A=V$.
Hence, $$\hat{V}-V=\hat{A}-A.$$From equations~\eqref{12}, \eqref{22} and the fact that server $2$ is honest it is clear that
\begin{equation}
\hat{A}-A=a\left(\hat{z}_1-z_2\right)=a\Delta_0\ne0
\end{equation}
\begin{equation}
\hat{V}-V=\tilde{a}\left(\hat{w}_1-w_2\right)=\tilde{a}\Delta_1,
\end{equation}
where $a=\frac{u_2}{u_2-u_1}$ and $\tilde{a}=\frac{\tilde{u}_2}{\tilde{u}_2-\tilde{u}_1}$.
Hence $\mathcal{A}$  wins the security  experiment $\textrm{EXP}_{\mathcal{A},\Pi_1}^{PriV}(m,\mathbf{x},i,j)$ iff it finds the solution for the equation
\begin{align}
&G(u_1, u_2, \tilde{u}_1,\tilde{u}_2)=\frac{\tilde{u}_2}{\tilde{u}_2-\tilde{u}_1}\Delta_1-\frac{u_2}{u_2-u_1}\Delta_0 \notag\\
&=\frac{\tilde{u}_2(u_2-u_1)\Delta_1-u_2(\tilde{u}_2-\tilde{u}_1)\Delta_0}{(\tilde{u}_2-\tilde{u}_1)(u_2-u_1)}=0.
\end{align}
From the description of security experiment, it follows that $\Delta_0, \Delta_1$ are known to $\mathcal{A}$ and independent from $u_1, u_2, \tilde{u}_1,\tilde{u}_2$. Hence,
\begin{comment}
	\begin{equation*}
		\textrm{Pr}[\textrm{EXP}_{\mathcal{A},\mathbb{A}}^{PriV}(m,\mathbf{x},i,j)=1]\leq
	\end{equation*}
	\begin{equation*}
		\textrm{Pr}(G(u_1, u_2, \tilde{u}_1,\tilde{u}_2)=0|\Delta_1\ne\Delta_0\ne0)
	\end{equation*}
	\begin{equation*}
		=\sum_{(c_1, c_2, c_3, c_4)\in L}\textrm{Pr}(G(c_1, c_2, c_3, c_4)=0|\Delta_1\ne\Delta_0\ne0)
	\end{equation*}
	\begin{equation}
		\cdot\textrm{Pr}[(u_1, u_2, \tilde{u}_1,\tilde{u}_2)=(c_1,c_2,c_3,c_4)]=\frac{g}{(q-1)^2(q-2)^2},    
	\end{equation}
\end{comment}
\begin{align}
&\textrm{Pr}[\textrm{EXP}_{\mathcal{A},\mathbb{A}}^{PriV}(m,\mathbf{x},i,j)=1] \notag \\
&\leq \textrm{Pr}(G(u_1, u_2, \tilde{u}_1,\tilde{u}_2)=0|\Delta_1\ne\Delta_0\ne0) \notag\\
&=\sum_{(c_1, c_2, c_3, c_4)\in L}\textrm{Pr}(G(c_1, c_2, c_3, c_4)=0|\Delta_1\ne\Delta_0\ne0)\notag\\
&~\cdot\textrm{Pr}[(u_1, u_2, \tilde{u}_1,\tilde{u}_2)=(c_1,c_2,c_3,c_4)]=\frac{g}{(q-1)^2(q-2)^2}, 
\end{align}
where $L=\{(c_1, c_2, c_3, c_4)|c_1, c_2, c_3, c_4\in\mathbb{F}_q\setminus\{0\}, c_1\ne c_2, c_3\ne c_4\}$
and $g$ is the number of $(c_1, c_2, c_3, c_4)\in L$ so that $G(c_1, c_2, c_3, c_4)=0$. Since the denominator of $G$ is non-zero, we can estimate $g$ as number of zeros of $H(u_1, u_2, \tilde{u}_1,\tilde{u}_2)=\tilde{u}_2(u_2-u_1)\Delta_1-u_2(\tilde{u}_2-\tilde{u}_1)\Delta_0$. Let $L'=\{(c_1, c_2, c_3, c_4)| c_1, c_2, c_3, c_4\in\mathbb{F}_q\setminus\{0\}\}$. It is clear that $L\subseteq L'$ and we can estimate $g$ from above as number of zeros of $H(\tau_1, \tau_2, \tau_3, \tau_4)$ where $\tau_1, \tau_2, \tau_3, \tau_4$ are chosen uniformly from $L'$. \textcolor{black}{Since $\Delta_0\ne0$}, by Schwartz-Zippel Lemma \cite{SZ1, SZ2} this value can be estimated from above as $\frac{2}{q-1}$. As a result, we have that $g\leq2(q-1)^3$, and the theorem statement follows.
\end{proof}
\balance
\section{Conclusion}
We considered the problem of reflecting malicious server behavior in private information retrieval schemes with optimized download rates. We focused on the extreme two-server case and propose generalizations of linear secret-sharing-based PIR schemes to verifiable setups. Considered schemes can detect the presence of one cheating server and offers information-theoretical private verifiability, computational public verifiability, and computational private verifiability with download rate close to those of non-verifiable scheme. \textcolor{black}{Extending the proposed framework to more than two servers or general linear PIR scheme are interesting open problems. Another possible research direction is an extension to accountability/Byzantine resistance.}

{\section*{Acknowledgements.}
This research/project is supported by the National Research Foundation, Singapore under its Strategic Capability Research Centres Funding Initiative, Singapore Ministry of Education Academic Research Fund Tier 2 Grants MOE2019-T2-2-083 and MOE-T2EP20121-0007, and the ARC grants DE180100768 and DP200100731. Any opinions, findings and conclusions or recommendations expressed in this material are those of the author(s) and do not reflect the views of National Research Foundation, Singapore.\\

Authors thank L.~F.~Zhang from ShanghaiTech University for numerous fruitful discussions during the preparation of this paper.}

\printbibliography
\end{document}